\documentclass[11pt]{article}

\usepackage{amssymb,amsmath}
\usepackage[english]{babel}
\usepackage{graphicx}
\newtheorem{DE}{Definition}[section]

\newcommand {\sm} {\setminus}

\usepackage{latexsym} 
\usepackage{theorem} 
\usepackage{comment}
\newcommand{\qed}{\relax\ifmmode\hskip2em\Box\else\unskip\nobreak\hfill$\Box$\fi}

\usepackage{alltt} 

\newtheorem{theorem}[DE]{Theorem}
\newtheorem{lemma}[DE]{Lemma}

{\theoremstyle{break}\theorembodyfont{\rmfamily}}
{\theoremstyle{break}\theorembodyfont{\rmfamily}}

 \newcounter{claim}

\newenvironment{claim}[1][]%
{\refstepcounter{claim}\vspace{1ex}\noindent{(\it\arabic{claim}){#1}{}}\it}{\vspace{1ex}}

\newenvironment{proofclaim}[1][]%
	{\noindent {}{#1}{}}{ This proves~(\arabic{claim}).\vspace{1ex}}

\newenvironment{proof}[1][]%
{\noindent {\setcounter{claim}{0}\em Proof. 
   }{#1}{}}{\hfill$\Box$\vspace{2ex}} 

\begin{document}

\title{Detecting wheels}

\author{Emilie Diot,  S\'ebastien Tavenas and Nicolas
  Trotignon\thanks{The affiliation of the authors and the grants
    supporting them are all given at the end of the paper.}}

\maketitle

\begin{abstract}
  A \emph{wheel} is a graph made of a cycle of length at least~4
  together with a vertex that has at least three neighbors in the
  cycle.  We prove that the problem whose instance is a graph $G$ and whose
  question is ``does $G$ contains a wheel as an induced subgraph'' is
  NP-complete.    We also settle the complexity of several similar
  problems. 
\end{abstract}

\section{Introduction}

In this article, all graphs are finite and simple.  If $G$ and $H$
are graphs, we say that $G$ \emph{contains} $H$ when $H$ is isomorphic
to an induced subgraph of $H$. 

A \emph{prism} is a graph made of three vertex-disjoint paths $P_1 =
a_1 \dots b_1$, $P_2 = a_2 \dots b_2$, $P_3 = a_3 \dots b_3$ of length
at least 1, such that $a_1a_2a_3$ and $b_1b_2b_3$ are triangles and no
edges exist between the paths except these of the two triangles.
 
A \emph{pyramid} is a graph made of three paths $P_1 = a \dots b_1$,
$P_2 = a \dots b_2$, $P_3 = a \dots b_3$ of length at least~1, two of
which have length at least 2, vertex-disjoint except at $a$, and such
that $b_1b_2b_3$ is a triangle and no edges exist between the paths
except these of the triangle and the three edges incident to $a$.

A \emph{theta} is a graph made of three internally vertex-disjoint
paths $P_1 = a \dots b$, $P_2 = a \dots b$, $P_3 = a \dots b$ of
length at least~2 and such that no edges exist between the paths
except the three edges incident to $a$ and the three edges incident to
$b$.

\begin{figure}
  \begin{center}
    \includegraphics[height=2cm]{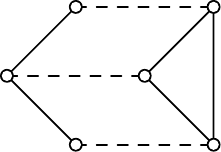}
    \hspace{.2em}
    \includegraphics[height=2cm]{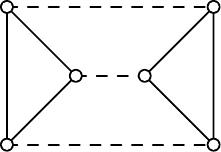}
    \hspace{.2em}
    \includegraphics[height=2cm]{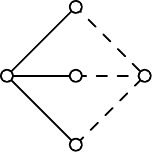}
    \hspace{.2em}
    \includegraphics[height=2cm]{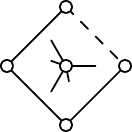}
  \end{center}
  \caption{Pyramid, prism, theta and wheel (dashed lines represent
    paths)\label{f:tc}}
\end{figure}

A \emph{hole} in a graph is a chordless cycle of length at least~4.
Observe that the lengths of the paths in the three definitions above
are designed so that the union of any two of the paths form a hole.  A
\emph{wheel} is a graph formed by a hole $H$ (called the \emph{rim})
together with a vertex (called the \emph{center}) that has at least
three neighbors in the hole.

A \emph{Truemper configuration} is a graph isomorphic to a prism, a
pyramid, a theta or a wheel (they were first considered by
Truemper~\cite{truemper}).  Truemper configurations play an important
role in the proof of several decomposition theorems as explained in a
very complete survey of Vu\v skovi\'c~\cite{vuskovic:truemper}.  Let
us explain how with the example of perfect graphs.

The \emph{chromatic number} of a graph $G$, denoted by $\chi(G)$, is
the minimum number of colours needed to assign a colour to each vertex
of $G$ in such a way that adjacent vertices receive different colours.
The \emph{clique number} of $G$, denoted by $\omega(G)$ is the maximum
number of pairwise adjacent vertices in $G$.  Every graph $G$ clearly
satisfies $\chi(G) \geq \omega(G)$, because the vertices of a clique
must receive different colours.  A graph $G$ is \emph{perfect} if
every induced subgraph $H$ of $G$ satisfies $\chi(H) = \omega(H)$.  A
chordless cycle of length $2k+1$, $k\geq 2$, satisfies $3 = \chi >
\omega = 2$, and its complement satisfies $k+1 = \chi > \omega = k$.
An \emph{antihole} is an induced subgraph $H$ of $G$, such that
$\overline{H}$ is hole of $\overline{G}$.  A hole (resp.\ an antihole)
is \emph{odd} or \emph{even} according to the number of its vertices
(that is equal to the number of its edges).  A graph is \emph{Berge}
if it does not contain an odd hole nor an odd antihole.  The
following, known as the \emph{strong perfect graph theorem}
(\emph{SPGT} for short), was conjectured by Berge~\cite{berge:61} in
the 1960s and was the object of much research until it was finally
proved in 2002 by Chudnovsky, Robertson, Seymour and
Thomas~\cite{chudnovsky.r.s.t:spgt} (since then, a shorter proof was
discovered by Chudnovsky and Seymour~\cite{chudnovsky.seymour:even}).

\begin{theorem}[{Chudnovsky, Robertson, Seymour and Thomas 2002}]
  \label{th:spgt}
  A graph is perfect if and only if it is Berge.
\end{theorem}

One direction is easy: every perfect graph is Berge, since as we
observed already odd holes and antiholes satisfy $\chi = \omega + 1$.
The proof of the converse is very long and relies on structural graph
theory.  The main step is a \emph{decomposition theorem}, not worth
stating here, asserting that every Berge graph is either in a
well-understood \emph{basic} class of perfect graphs, or has some
\emph{decomposition}.

Let us now explain why Truemper configurations play a role in the
proof.  First, a Berge graph has no pyramid (because among the three
paths of a pyramid, two have the same parity, and their union forms an
odd hole).  This little fact is used very often to provide a
contradiction when describing the structure of a Berge graph.  A long
part of the proof of the SPGT is devoted to study the structure of a
Berge graph that contains a prism, and another long part is devoted to
a Berge graph that contains a wheel.  And at the very end of the
proof, it is proved that graphs not previously decomposed are
bipartite, just as Berge thetas are.  Note that prisms can be defined
as line graphs of thetas.  This use of Truemper configurations is
seemingly something deep and general as suggested by the survey of
Vu\v skovi\'c~\cite{vuskovic:truemper} about Truemper configurations
and how they are used (sometimes implicitly) in many decomposition
theorems.

\vspace{2ex}

Testing whether a graph contains or not  Truemper
configurations is therefore a question of interest.  In what follows,
$n$ stands for the number of vertices, and $m$ for the number of edges
of the input graph.  Detecting a pyramid in an input graph can be
done in time $O(n^9)$ (see Chudnovsky, Cornu\'ejols, Liu, Seymour and
Vu\v skovi\'c~\cite{chudnovsky.c.l.s.v:reco}) and a theta in time
$O(n^{11})$ (see Chudnovsky and
Seymour~\cite{chudnovsky.seymour:theta}).  Detecting a prism is
NP-complete (see Maffray and Trotignon~\cite{maffray.t:reco}).
Detecting a prism or a pyramid can be done in time $O(n^5)$ (see
Maffray and Trotignon~\cite{maffray.t:reco}).  Detecting a theta or a
pyramid can be done in time $O(n^7)$ (see Maffray, Trotignon and Vu\v
skovi\'c~\cite{maffray.t.v:3pcsquare}).  Detecting a prism or a theta
can be done in time $O(n^{35})$ (see Chudnovsky and
Kapadia~\cite{Chudnovsky.Ka:08}).

The complexity of detecting a wheel was not known so far.  We prove
here that it is NP-complete, even when restricted to bipartite (and
therefore perfect) graphs.  Our proof relies on a variant of a
classical construction of Bienstock~\cite{bienstock:evenpair} (that is
the basis of all the hardness results in the field, but how to use it
for wheels has not been discovered so far).  An easy consequence is
that detecting a wheel or a prism is NP-complete (because bipartite
graphs contain no prisms, so for them detecting a wheel or a prism is
equivalent to detecting a wheel). By the same argument, detecting a
wheel or a pyramid is NP-complete.  Also detecting a wheel, a pyramid
or a prism is NP-complete.

\newcommand{\zero}{yes} \newcommand{\un}{---}

In Table~\ref{t:t}, we survey the complexity of detecting any
combination of Truemper configuration.  The structure to be detected
is indicated with ``\zero''.  For instance line~5 of the table should
be read as follows: the complexity of deciding whether a graph
contains a theta or a prism is $O(n^{35})$.  Observe that being able
to detect a theta \emph{or} a prism is equivalent to a recognition
algorithm of the classes of graphs that do not contain thetas
\emph{and} prisms as induced subgraphs.  Line~0 of the table follows
from a result of Conforti, Cornu{\'e}jols, Kapoor and Vu{\v s}kovi{\'
  c}~\cite{confortiCKV97}.  They call \emph{universally signable
  graphs} the graphs that contain no Truemper configuration, and give
a decomposition theorem for them.  The complexity of recognizing
universally signable graphs is obtained with an algorithm of
Tarjan~\cite{tarjan:clique} that gives the decomposition tree of any
graph with clique cutsets.  For lines with a question mark, the
complexity is not known.  The complexities claimed in lines 8, 10, 12
and 14 of the table are proved in this article.

\begin{table}
  \begin{tabular}{ccccccc}
    k & theta & pyramid & prism & wheel & Complexity & Reference\\ \hline
    0   &  \zero  & \zero & \zero & \zero & $O(nm)$&\cite{confortiCKV97}\cite{tarjan:clique}\\
    1   &  \zero  & \zero & \zero & \un & $O(n^7)$& \cite{maffray.t:reco}\cite{maffray.t.v:3pcsquare} \\
    2   &  \zero  & \zero & \un & \zero & ? & \\
    3   &  \zero  & \zero & \un & \un & $O(n^7)$&\cite{maffray.t.v:3pcsquare}\\
    4   &  \zero  & \un & \zero & \zero & ?& \\
    5   &  \zero  & \un & \zero & \un & $O(n^{35})$&\cite{Chudnovsky.Ka:08} \\
    6   &  \zero  & \un & \un & \zero & ?& \\
    7   &  \zero  & \un & \un & \un & $O(n^{11})$&\cite{chudnovsky.seymour:theta}\\
    8   &  \un  & \zero & \zero & \zero & NPC& \\
    9   &  \un  & \zero & \zero & \un & $O(n^5)$&\cite{maffray.t:reco} \\
    10   &  \un  & \zero & \un & \zero & NPC& \\
    11   &  \un  & \zero & \un & \un & $O(n^9)$&\cite{chudnovsky.c.l.s.v:reco} \\
    12   &  \un  & \un & \zero & \zero & NPC&\\
    13   &  \un  & \un & \zero & \un & NPC&\cite{maffray.t:reco} \\
    14   &  \un  & \un & \un & \zero & NPC&\\
    15   &  \un  & \un & \un & \un & $O(1)$&\\
  \end{tabular}
  \caption{Detecting Truemper configurations\label{t:t}}
\end{table}

In Section~\ref{sec:base} we give the basic reduction from 3-SAT that
is used for all our hardness results.  In Section~\ref{sec:wheels}, we
adress the question of detecting a wheel, and several variants
motivated by perfect graphs, such as detecting a wheel in a graph or
its complements, and detecting variants of wheels (with different sets
of constraints on the length of the rim, and the numbers of neighbors
of the center).  Some variants are polynomial, and
some are NP-complete.

\section{The main construction}
\label{sec:base}

In this section, we give a variant of a classical construction due to
Bienstock.  Let $f$ be an instance of $3$-SAT, consisting of $m$
clauses $C_1, \ldots, C_m$ on $n$ variables $x_1, \ldots, x_n$.  Let
us build a graph $G_f$ with two specialized vertices $a,b$, such that
there will be an induced cycle containing both $a,b$ in $G$ if and
only if there exists a truth assignment for $f$.  For later use, some
edges of $G_f$ will be labelled ``black'' and some will be labelled
``red''.  Black edges should be thought of as ``edges that can be
subdivided'', or as ``edges that potentially belong to the hole''.
Red edges should be thought of as ``edges that serve as chords'' and
as ``non-subdivisible edges''.

For each variable $x_i$ ($i=1, \ldots, n$), make a graph $G(x_i)$ with
$8m+8$ vertices $a_i, b_i, a'_i, b'_i, t_{i, 0}, \dots, t_{i, 2m},
f_{i, 0}, \dots, f_{i, 2m}, t'_{i, 0}, \dots, t'_{i, 2m}, f'_{i, 0},
\dots, f'_{i, 2m},$.  Add black edges in such a way that $a_i t_{i, 0}
\dots t_{i, 2m} b_i$, $a_i f_{i, 0} \dots f_{i, 2m} b_i$, $a'_i t'_{i,
  0} \dots t'_{i, 2m} b'_i$ and $a'_i f'_{i, 0} \dots f'_{i, 2m} b'_i$
are chordless paths.  Add the following red edges: $t_{i, 2j} f_{i,
  2j}$, $f_{i, 2j} t'_{i, 2j}$, $t'_{i, 2j} f'_{2j}$ and $f'_{i, 2j}
t_{i, 2j}$ for $j=0, \dots, m$.  See Figure~\ref{fig:gxi}.

For each clause $C_j$ ($j=1, \ldots, m$), with $C_j=u_j^1\vee
u_j^2\vee u_j^3$, where each $u_j^p$ ($p=1, 2, 3$) is a literal from
$\{x_1, \ldots, x_n, \overline{x}_1, \ldots, \overline{x}_n\}$, make a
graph $G(C_j)$ with five vertices $c_j, d_j, v_j^1, v_j^2, v_j^3$ and
six black edges so that each of $c_j, d_j$ is adjacent to each of
$v_j^1, v_j^2, v_j^3$.  See Figure~\ref{fig:gcj}.  For $p=1, 2, 3$, if
$u_j^p=x_i$ then add two red edges $v_j^pf_{i, 2j-1}, v_j^pf'_{i,
  2j-1}$, while if $u_j^p=\overline{x}_i$ then add two red edges
$v_j^pt_{i, 2j-1}, v_j^pt'_{i, 2j-1}$.  See Figure~\ref{fig:gf}.

The graph $G_f$ is obtained from the disjoint union of the $G(x_i)$'s
and the $G(C_j)$'s as follows.  For $i=1, \ldots, n-1$, add edges
$b_ia_{i+1}$ and $b'_ia'_{i+1}$.  Add an edge $b'_nc_1$.  For $j=1,
\ldots, m-1$, add an edge $d_jc_{j+1}$.  Introduce the two specialized
vertices $a,b$ and add edges $aa_1, aa'_1$ and $bd_m, bb_n$.  See
Figure~\ref{fig:gf2}.  Clearly the size of $G_f$ is polynomial
(actually quadratic) in the size $n+m$ of $f$.  An \emph{$f$-graph} is
any graph obtained from $G_f$ by subdividing black edges of $G_f$ (the
subdivision is arbitrary: each edge is subdivided an arbitrary number
of times, possibly zero).

\begin{figure}
  \begin{center}
    \includegraphics{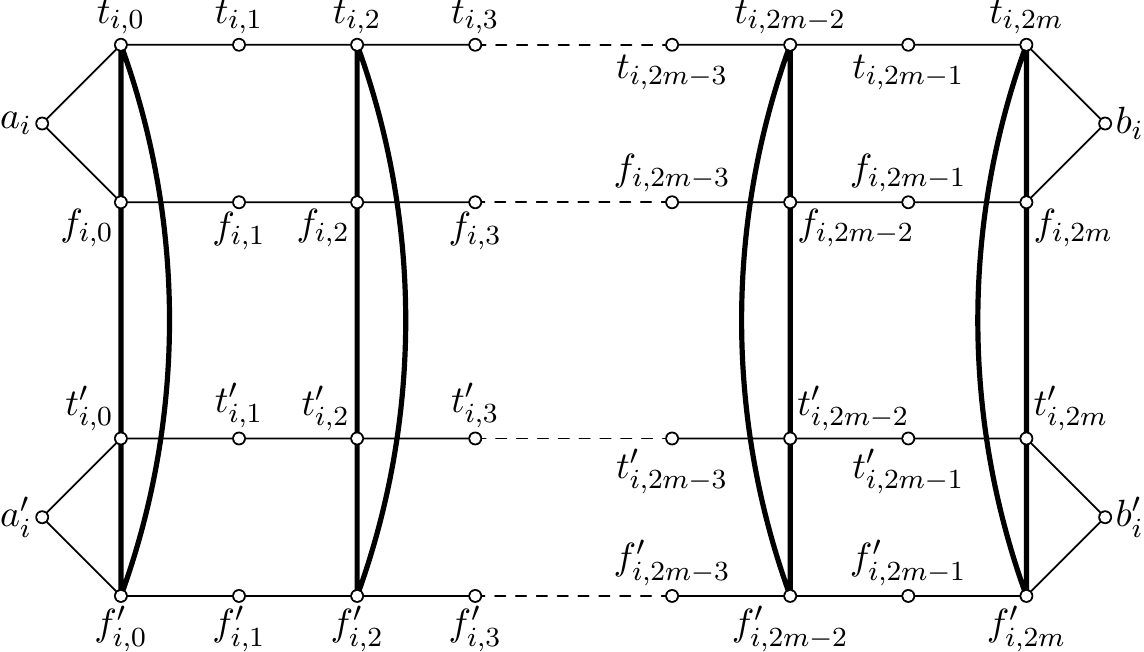}
    \caption{The graph $G(x_i)$\label{fig:gxi}}
  \end{center}
\end{figure}

\begin{figure}
  \begin{center}
    \includegraphics{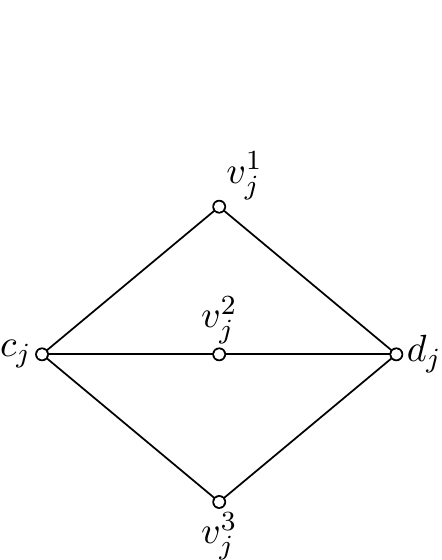}
    \caption{The graph $G(C_j)$\label{fig:gcj}}
  \end{center}
\end{figure}

\begin{figure}
  \begin{center}
    \includegraphics{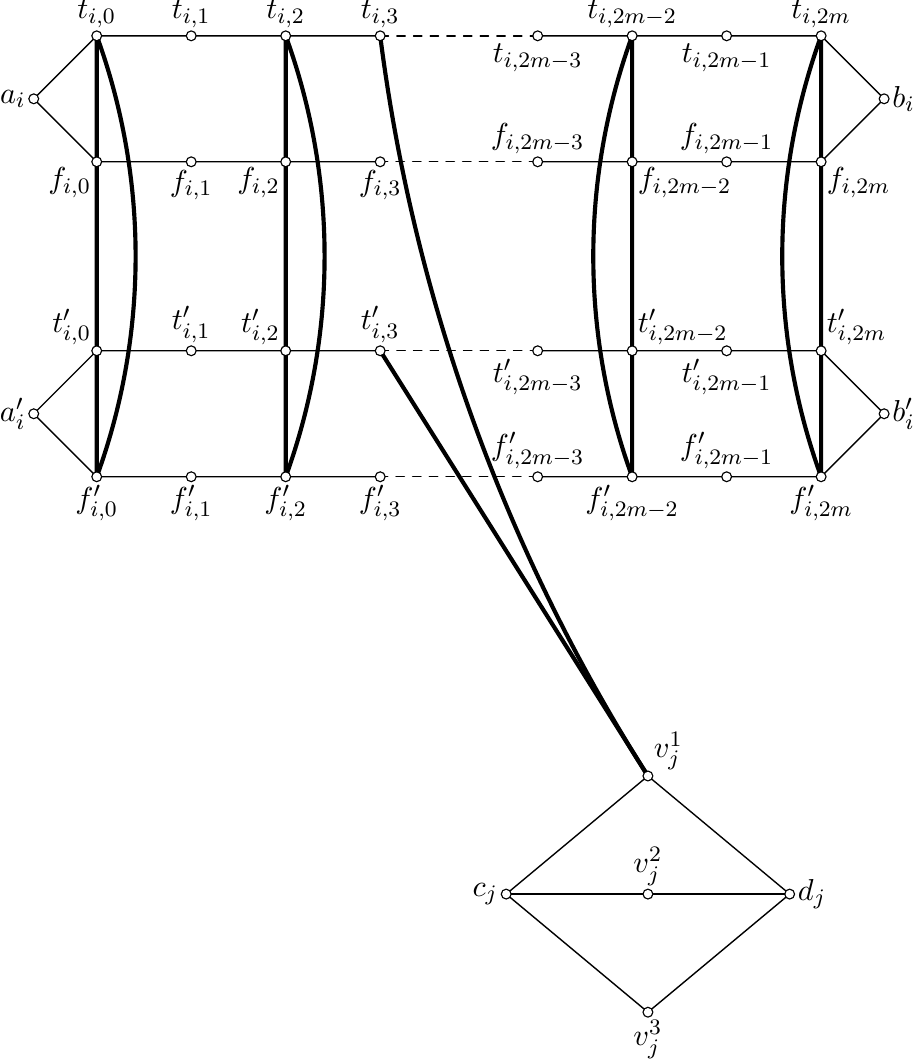}
    \caption{Red edges between $G(x_i)$ and $G(C_j)$ (here
      $j=2)$\label{fig:gf}}
  \end{center}
\end{figure}

\begin{figure}
  \begin{center}
    \includegraphics{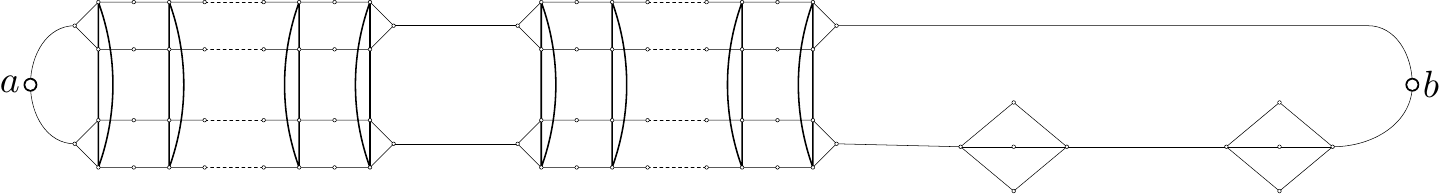}
    \caption{The graph $G_f$\label{fig:gf2}}
  \end{center}
\end{figure}

\begin{lemma}
  \label{l:holeab}
  Let $f$ be an instance of {$3$-SAT} and $G$ be an $f$-graph.  Then,
  $f$ admits a truth assignment if and only if $G$ contains an induced
  cycle through $a$ and $b$.
\end{lemma}

\begin{proof}
  Suppose that $f$ admits a truth assignment $\xi\in\{0, 1\}^n$.  We
  build an induced cycle $H$ in $G_f$ by selecting vertices as follows
  (we build an induced cycle in $G$ later).  Select $a,b$.  For $i=1,
  \ldots, n$, select $a_i, b_i, a'_i, b'_i$; moreover, if $\xi_i=1$
  select $t_{i, 0} \dots t_{i, 2m}$ and $t'_{i, 0} \dots t'_{i, 2m}$,
  while if $\xi_i=0$ select $f_{i, 0} \dots f_{i, 2m}$ and $f'_{i, 0}
  \dots f'_{i, 2m}$.  For $j=1, \ldots, m$, since $\xi$ is a truth
  assignment for $f$, at least one of the three literals of $C_j$ is
  equal to $1$, say $u_j^p=1$ for some $p\in\{1, 2, 3\}$.  Then select
  $c_j, d_j$ and $v_j^p$.  Now it is a routine matter to check that
  the selected vertices induce a cycle $Z$ that contains $a,b$, and
  that $Z$ is chordless, so it is an induced cycle.  The main point is
  that there is no chord in $Z$ between some subgraph $G(C_j)$ and
  some subgraph $G(x_i)$, for that would be either an edge $t_iv_j^p$
  (or $t'_iv_j^p$) with $u_j^p=x_i$ and $\xi_i=1$, or, symmetrically,
  an edge $f_iv_j^p$ (or $f'_iv_j^p$) with $u_j^p=\overline{x}_i$ and
  $\xi_i=0$, in either case a contradiction to the way the vertices of
  $Z$ were selected.

  To build an induced cycle in $G$ (instead of $G_f$), just subdivide
  the black edges of $E(G_f) \cap E(H)$ that were subdivided to obtain
  $G$.

  For the converse statement, we write the proof for $G_f$, the proof
  is the similar for $G$.  Suppose that $G_f$ admits an induced cycle
  $Z$ that contains $a,b$.  Clearly $Z$ contains $a_1, a'_1$ since
  these are the only neighbours of $a$ in $G_f$.

\begin{claim}\label{clm:zgxi}
  For $i=1, \ldots, n$, $Z$ contains exactly $4m+6$ vertices of
  $G(x_i)$: $a_i, a'_i, b_i, b'_i$, and either $t_{i, 0} \dots t_{i,
    2m}$ and $t'_{i, 0} \dots t'_{i, 2m}$, or $f_{i, 0} \dots f_{i,
    2m}$ and $f'_{i, 0} \dots f'_{i, 2m}$.
\end{claim}

\begin{proofclaim} 
  First we prove the claim for $i=1$.  Since $a, a_1$ are in $Z$ and
  $a_1$ has only three neighbours (namely $a, t_{1, 0}, f_{1, 0}$)
  exactly one of $t_{1, 0}, f_{1, 0}$ is in $Z$.  Likewise exactly one
  of $t'_{1, 0}, f'_{1, 0}$ is in $Z$.  If $t_{1, 0}, f'_{1, 0}$ are
  in $Z$, then the vertices $a, a_1, a'_1, t_{1, 0}, f'_{1, 0}$ are
  all in $Z$ and because of the red edge, they induce a cycle that
  does not contain $b$, a contradiction.  Likewise we do not have both
  $t'_{1, 0}, f_{1, 0}$ in $Z$.  Therefore, up to symmetry we may
  assume that $t_{1, 0}, t'_{1, 0}$ are in $Z$.  Thus $f_{1, 0}$ is
  not in $Z$ because of the red edge $t_{1, 0} f_{1, 0}$.  Similarly,
  $f'_{1, 0}$ is not in $Z$.  It follows that $t_{1, 1}$ and $t'_{1,
    1}$ are in $Z$.

  If a vertex $u_j^p$ of some $G(C_j)$ ($1\le j\le m$, $1\le p\le 3$)
  is in $Z$ and is adjacent to $t_{1, 1}$ then, since this $u_j^p$ is
  also adjacent to $t'_{1, 1}$, we see that the vertices $a, a_1,
  a'_1, t_{1, 0}, t'_{1, 0}, t_{1, 1}, t'_{1, 1}$ and $u_j^p$ are all
  in $Z$ and induce a hole that does not contain $b$, a contradiction.
  Thus the neighbour of $t_{1, 1}$ in $Z\setminus t_{1, 0}$ is not in
  any $G(C_j)$ ($1\le j\le m$), so that neighbour is $t_{1, 2}$.
  Likewise $t'_{1, 2}$ is in $Z$.  By the same argument, it can be
  proved that $t_{1, 3}$, \dots, $t_{1, 2m}$ and $t'_{1, 3}$, \dots,
  $t'_{1, 2m}$ are all in $Z$.  Also, for $k=1, \dots, m$, $f_{1, 2k}$
  is not in $Z$ because of the red edge $t_{1, 2k} f_{1, 2k}$ and
  similarly, $f'_{1, 2k}$ is not in $Z$.  Since $f_{1, 2k-1}$ has
  degree at most~3, it cannot be in $Z$ because one of its neighbor in
  $Z$ would be $f_{1, 2k-2}$ or $f_{1, 2k}$.  It follows that $b_1$
  and $b'_1$ are in $Z$.
  
  So the claim holds for $i=1$.  Since $f_{1, 2m}$ is not in $Z$, we
  see that $a_2$ is in $Z$ and similarly that $a'_2$ is in $Z$.  Now
  the proof of the claim for $i=2$ is essentially the same as for
  $i=1$, and so on up to $i=n$.
\end{proofclaim}

\begin{claim}\label{clm:nored}
  $Z$ contains no red edge of $G_f$.
\end{claim}

\begin{proofclaim}
  Follows directly from~(\ref{clm:zgxi}).
\end{proofclaim}

\begin{claim}\label{clm:zgcj}
  For $j=1, \ldots, m$, $Z$ contains $c_j, d_j$ and exactly one of
  $v_j^1, v_j^2, v_j^3$.
\end{claim}

\begin{proofclaim}
  First we prove this claim for $j=1$.  By~(\ref{clm:zgxi}), $b'_n$ is
  in $Z$ and exactly one of $t'_{n, 2m}, f'_{n, 2m}$ is in $Z$, so
  (since $b'_n$ has degree $3$ in $G_f$) $c_1$ is in $Z$. So,
  by~(\ref{clm:nored}) $Z$ contains exactly one of the paths
  $c_1v_1^1d_1$ , $c_1v_1^2d_1$ or $c_1v_1^3d_1$.  Thus, the neighbor
  of $d_1$ in $Z \sm \{v_1^1, v_1^2, v_1^3\}$ must be $c_2$. Now the
  proof of the claim for $j=2$ is the same as for $j=1$, and similarly
  the claim holds up to $j=m$.
\end{proofclaim}

We can now make a Boolean vector $\xi$ as follows.  For $i=1, \ldots,
n$, if $Z$ contains $t_{i, 0}, t'_{i, 0}$ set $\xi_i = 1$; if $Z$
contains $f_{i, 0}, f'_{i, 0}$ set $\xi_i = 0$.  By~(\ref{clm:zgxi})
this is consistent.  Consider any clause $C_j$ ($1\le j\le m$).
By~(\ref{clm:zgcj}) and up to symmetry we may assume that $v_j^1$ is
in $Z$.  If $v_j^1 = x_i$ for some $i\in\{1, .., n\}$, then the
construction of $G_f$ implies that $f_{i, 2j-1}, f'_{i, 2j-1}$ are not
in $Z$, so $t_{i, 0}, t'_{i, 0}$ are in $Z$, so $\xi_i=1$, so clause
$C_j$ is satisfied by $x_i$.  If $v_j^1 = \overline{x}_i$ for some
$i\in\{1, .., n\}$, then the construction of $G_f$ implies that $t_{i,
  2j-1}, t'_{i, 2j-1}$ are not in $Z$, so $f_{i, 0}, f'_{i, 0}$ are in
$Z$, so $\xi_i=0$, so clause $C_j$ is satisfied by $\overline{x}_i$.
Thus $\xi$ is a truth assignment for $f$.  This completes the proof of
the lemma.
\end{proof}

A \emph{hub} in a graph is a vertex that has at least three neighbors
of degree at least 3.  Note that the center of a wheel is a hub.  This
simple observation will be useful in the next section.

\begin{theorem}
  \label{th:holeBipNoHub}
  Let $k$ be an integer.  The problem of detecting an induced cycle of
  length at least $k$ through two prescribed vertices $a$ and $b$ of
  degree 2 in an input graph is NP-complete, even when restricted to
  bipartite graphs with no hub.
\end{theorem}

\begin{proof}
  We consider an instance $f$ of 3-SAT and we build an instance $(G,
  a, b)$ of our problem such that $f$ can be satisfied if an only if
  an induced cycle of length at least $k$ of $G$ goes through $a$ and
  $b$.
  
  To do so, we apply Lemma~\ref{l:holeab} and we start with the graph
  $G_f$ defined above.  We now subdivide carefully chosen black edges.
  First, we subdivide the edge $aa_1$ $k$ times so that any cycle
  through $a$ has length at least $k$.  Since every vertex of $G$ is
  adjacent to at most 2 red edges, it is possible to eliminate all
  hubs by subdividing once each black edge.  Now consider the vertices
  of degree at least~3 in $G$.  They induce a graph $G_{\geq 3}$.  The
  components of $G_{\geq 3}$ are cycles of length~4, paths on 3
  vertices and isolated vertices.  Thus $G_{\geq 3}$ is bipartite, and
  we choose a bipartition into blue and green vertices (it is not
  unique since $G_{\geq 3}$ is not connected).  Now the fact that the
  bipartition of $G_{\geq 3}$ can be extended to a bipartition of $G$
  depends only on the parity of the paths of black edges linking the
  components of $G_{\geq 3}$.  It follows that by subviding one edge
  or no edges in each of these paths, a bipartite graph can be
  obtained.
\end{proof}

\section{Detecting wheels}
\label{sec:wheels}

Let $3\leq k$ and $0 \leq l \leq k$ be integers.  A \emph{$(k,
  l)$-wheel} is a graph made of a chordless cycle of length at least
$k$ together with a vertex that has at least $l$ neighbors in the
cycle.  Thus a wheel is a $(4, 3)$-wheel.  An article of Aboulker et
al.~\cite{aboulkerRTV:propeller} is devoted to the detection of $(3,
2)$-wheels (that are called \emph{propellers}).  The last fifty pages
of the proof of the SPGT are devoted to studying Berge graphs that
contain particular kinds of $(6, 3)$-wheels or their complement.
These examples are our motivation for not restricting ourselves to
the study of wheels.

If $3 \leq k$ and $3 \leq l \leq k$ are integers, then the center of a
$(k, l)$-wheel is a hub.  This simple observation is very useful to
transform Lemma~\ref{l:holeab} into the next theorem.

We denote by $\Pi_{k, l}$ the problem whose instance is a graph $G$
and whose question is ``does $G$ contain a $(k, l)$-wheel as an
induced subgraph?''.  We do not know the complexity of $\Pi_{k, l}$
when $l = 2$ and $k \geq 4$ (for large values of $k$, we believe that
the question is quite chalenging).  The next theorem settles all the
other cases.

\begin{figure}
  \begin{center}
    \includegraphics{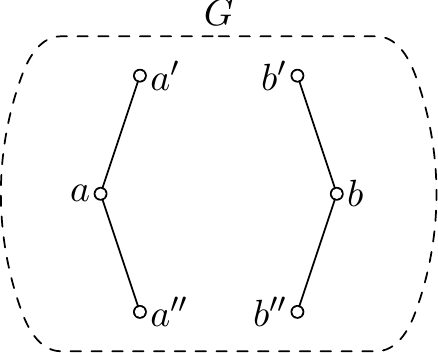}
    \rule{1em}{0ex}\includegraphics{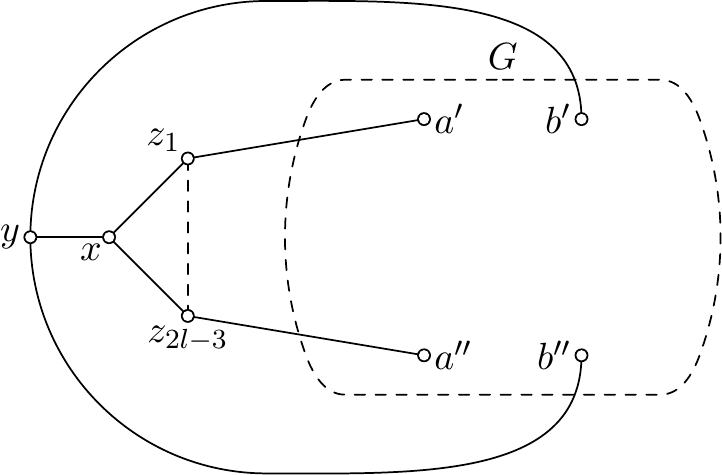}
    \caption{$G$ and $G'$\label{fig:wheel}}
  \end{center}
\end{figure}

\begin{theorem}
  Let $3 \leq k$ and $0 \leq l \leq k$ be integers.  The problem
  $\Pi_{k, l}$ is polynomial if $l\leq 1$ and is NP-complete if $l
  \geq 3$ (it remains NP-complete when restricted to bipartite
  instances).  If $l=2$ and $k=3$ then $\Pi_{k, l}$ is polynomial.
\end{theorem}

\begin{proof}
  If $l = 0$, then $\Pi_{k, l}$ consists in detecting an induced cycle
  of length at least $k$, and if $l = 1$ it can be reduced easily to
  the detection of an induced cycle of length at least $k$ through a
  prescribed vertex.  These problems are clearly polynomial.  If $l=2$
  and $k=3$, then the polynomiality of $\Pi_{k, l}$ is a result
  in~\cite{aboulkerRTV:propeller} where it is refered to as the
  \emph{detection of propellers}.

  Let us now study the NP-complete cases, so suppose $l\geq 3$.
  Consider an instance $(G, a, b)$ of the problem from
  Theorem~\ref{th:holeBipNoHub} (so $G$ is bipartite).  Call $a', a''$
  the two neighbors of $a$ and $b',b''$ the two neighbors of
  $b$. Build a graph $G'$ as follows.  Delete $a$ and $b$ from $G$.
  Add a path $z_1 \dots z_{2l-3}$ and the edges $z_1 a'$ and
  $z_{2l-3}a''$.  Add a vertex $x$ adjacent to $z_1, z_3, \dots,
  z_{2l-3}$.  Add a vertex $y$ adjacent to $x$, $b'$ and $b''$ (see
  Fig.~\ref{fig:wheel}). After possibly subdividing once or twice
  $z_1a'$, $z_{2l-3}a''$, $yb'$ and $yb''$, $G'$ is bipartite and the
  unique hub of $G'$ is~$x$.

  A $(k, l)$-wheel of $G'$ must therefore be centered at $x$ and must
  contain an induced path from $a'$ to $b'$ and an induced path from
  $a''$ to $b''$ (or an induced path from $a'$ to $b''$ and an induced
  path from $a''$ to $b'$).  In either case, $G$ contains a hole that
  goes through $a$ and $b$.

  Conversely, if $G$ contains a hole through $a$ and $b$, then $G'$
  contains a $(k, l)$-wheel (centered at $x$).
\end{proof}

The class of Berge graphs is self-complementary, and detecting a
structure in a Berge graph or its complement is sometimes useful.
This motivates the next problem.  We denote by $\overline{\Pi}_{k, l}$
the problem whose instance is a bipartite graph $G$ and whose question
is ``does one of $G$ or $\overline{G}$ contain a $(k, l)$-wheel as an
induced subgraph?''.  The next theorem settles the complexity of
$\overline{\Pi}_{k, l}$ in several cases (the other cases are open).

\begin{theorem}
  Let $3 \leq k$ and $0 \leq l \leq k$ be integers.  The problem
  $\overline{\Pi}_{k, l}$ is polynomial if $k\leq 4$ and NP-complete
  if $k\geq 5$ and $l\geq 3$ (it remains NP-complete when restricted
  to bipartite instances).
\end{theorem}

\begin{proof}
  Suppose that $k \leq 4$ (so $l\leq 4$).  Consider a chordless path
  $v_1 \dots v_7$.  In the complement, the vertices $v_1, v_2, v_4,
  v_5, v_7$ induce a $(4, 4)$-wheel (and therefore a $(k, l)$-wheel
  for any $k\leq 4$ and $0 \leq l \leq k$).  Since a $(k, l)$-wheel on
  at least nine vertices contains an induced cycle of length 8, it
  also contains an induced path of length 7, and therefore a $(k,
  l)$-wheel on 5 vertices in the complement.  It follows that the
  answer to $\overline{\Pi}_{k, l}$ is yes if and only if the input
  graph or its complement contains a wheel on at most eight vertices.
  This can be tested by brute force enumeration in time $O(n^8)$.

  When $k\geq 5$, the complement of a bipartite graph cannot contain a
  $(k, l)$-wheel.  Because $C_5$ is self-complementary and a cycle of
  length at least~6 contains a stable set of size 3 (that does not
  exist in the complement of a bipartite graph).  Thus, the
  NP-completeness of $\overline{\Pi}_{k, l}$ directly follows from
  NP-completeness of ${\Pi}_{k, l}$.
\end{proof}

In fact, we can be faster for $\overline{\Pi}_{4, 3}$, that is the
problem of detecting a wheel in an input graph or in its
complement. We will need the following result which appears as
Theorem~3.1 in Nikolopoulos and Palios' paper~\cite{nikolopoulos.s}.

\begin{theorem}\label{Prop_NPAlgorithm}
  There is an $O(n + m^2)$ time algorithm that determines whether an input
  graph $G$ contains a hole of size at least five.
\end{theorem}

We also need the next two little facts.

\begin{lemma}
  \label{l:h}
  If a graph $G$ contains a hole $H$ of length at least 5,  then either
  $G=H$ or one of  $G$ or $\overline{G}$ contains a wheel.
\end{lemma}
 
 \begin{proof} 
   Suppose that $G\neq H$ and let $w\notin V(H)$ be a vertex of $G$.
   If $H$ is of length~5, then it is self-complementary, so in $G$ or
   $\overline{G}$, $w$ has at least three neighbors in $H$.  If $H$ is
   of length at least~6, then let $H = v_1 \dots v_6 \dots$.  If $G$
   has no wheel, then $w$ has at most two neighbors in $H$.  Up to a
   relabelling, we may assume that $w$ has at most one neighbor among
   $v_2, \dots, v_6$. It follows that $\{w, v_2, v_3, v_5, v_6\}$
   induces a wheel in $\overline{G}$.
\end{proof}

\begin{lemma}
  \label{l:p3p3bar}
  Let $G$ be a graph.  None of $G, \overline{G}$ contains a path on
  three vertices if and only if   $G$ is a complete graph  or  an
  independent graph. 
\end{lemma}

\begin{proof}
  Suppose that $G$ is neither independent nor complete.  Since $G$ is
  not independent, it has a connected component $C$ with at least two
  vertices.  So, $C$ contain an edge $vw$.  If $G$ has another
  connected component $C'$, then for some $u\in C'$, $vuw$ is a path
  on three vertices in $\overline{G}$.  So, we may assume that $C$ is
  the only connected component of $G$.  Since $C=G$ is not complete,
  it contains two non-adjacent vertices $v', w'$.  A shortest path from $v'$ to
  $w'$ in $G$ contains a path on three vertices.  We proved that $G$
  contains a path on three vertices or its complement.  The proof of
  the converse statement is clear.
\end{proof}

  \begin{figure}
    \begin{center}
    \includegraphics[height=2.5cm]{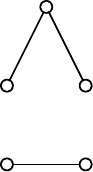}\rule{2cm}{0cm}
    \includegraphics[height=2.5cm]{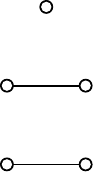}
  \end{center}
  \caption{The two complements of wheels on five vertices\label{fig:ComplementWheel5}}
\end{figure}

\begin{theorem}\label{Thm_AlgoPi43}
  The problem $\overline{\Pi}_{4, 3}$ can be solved in time $O(n^4)$.
\end{theorem}

\begin{proof}
  The first step of the algorithm  is to check whether $G$ or
$\overline{G}$ contains a hole of length at least~5.  This can be
implemented in time $O(n^4)$ by Theorem~\ref{Prop_NPAlgorithm}, and if
a hole is found, Lemma~\ref{l:h} allows to decide easily whether $G$ or
$\overline{G}$ contains a wheel.  So, we may assume from here on that
none of $G$, $\overline{G}$ contains a hole of length at least~5.
Hence, we just need to detect a wheel on five vertices.  For
convenience, we show how to detect the complement of a wheel (and we run
this routine in the graph and in its complement).

Complement of wheels on five vertices are represented in
Fig.~\ref{fig:ComplementWheel5}.  One is the disjoint union of an edge
and a path on three vertices, the other one is the disjoint union of
an edge and the complement of a path on three vertices.  To decide
whether a graph contains the complement of a wheel on five vertices,
it is therefore enough to check all edges $vw$ of $G$, and to decide
for each of them whether the set $S$ of vertices of $G$ adjacent to
none of $v, w$ contains a path on three vertices or its complement.
By Lemma~\ref{l:p3p3bar}, testing the desired property in $S$ is easy
to implement in time $O(n^2)$.  Hence, the algorithm can be
implemented to run in time $O(n^4)$.
\end{proof}   

Theorem~\ref{Thm_AlgoPi43} and its proof suggest that graphs with no
wheels and no complement of wheels form a restricted class that might
have a simple structure.  The class contains all split graphs,
complete bipartite graphs, some (non-induced) subgraphs of them, and
several particular graphs such as $C_5$, $C_6$ or $P_5$.  We could not
ellucidate its structure, and leave this as an open question.


\section*{Affiliation and grants}

\begin{itemize}
\item Emilie Diot, S\'ebastien Tavenas and Nicolas Trotignon: CNRS, ENS de
  Lyon, LIP, \'Equipe MC2, INRIA, Universit\'e Lyon 1, Universit\'e de
  Lyon. 
\item This work was supported by the LABEX MILYON (ANR-10-LABX-0070)
  of Universit\'e de Lyon, within the program "Investissements d'Avenir"
  (ANR-11-IDEX-0007) operated by the French National Research Agency
  (ANR). 
\item Nicolas Trotignon is partially supported by the French Agence Nationale de la
      Recherche under reference \textsc{anr-10-jcjc-Heredia} and
      \textsc{anr-14-blan-stint.}   

\end{itemize}

\end{document}